\documentclass{article}

\usepackage[a4paper, total={6in, 8in}]{geometry}

\usepackage{amsmath}

\usepackage{amsthm}

\usepackage{natbib}

\usepackage[plain,noend]{algorithm2e}

\usepackage{authblk}

\usepackage{amssymb}            
\usepackage{graphicx}

\newcommand{\infig}[2]{  \includegraphics[width=#1\textwidth]{#2}  }

\newcommand{\mycomments}[1]{}

 \newcommand{\sfrac}[2]{ {\textstyle\frac{#1}{#2}}}

\let\LaTeXtitle\title
\renewcommand{\title}[1]{\LaTeXtitle{\large\textsf{\textbf{#1}}}}

\newtheorem{proposition}{Proposition}
\newcommand{\tbl}{\caption}

\title{Minimum reversion in multivariate time series}

\author[1]{T. KLEINOW}
\author[2]{M.H. VELLEKOOP}
\affil[1]{Department of Actuarial Mathematics and Statistics
       and the Maxwell Institute for Mathematical Sciences,
    School of Mathematical and Computer Sciences,      Heriot-Watt University, EH14 4AS, Edinburgh, U.K. Email: {\tt t.kleinow@hw.ac.uk}}
\affil[2]{
Amsterdam School of Economics,
Faculty of Economics \& Business, University of Amsterdam, 1018 WB Amsterdam, the Netherlands.\\ Email: {\tt m.h.vellekoop@uva.nl}}

\makeatletter
\def\blfootnote{\xdef\@thefnmark{}\@footnotetext}
\makeatother

\begin{document}

  \maketitle

\blfootnote{\noindent
Torsten Kleinow acknowledges financial support from the Actuarial Research Centre of the Institute and Faculty of Actuaries, through the research programme on {\sl Modelling, Measurement and Management of Longevity and Morbidity Risk}. 
}

\begin{abstract}
\noindent 
We propose a new multivariate time series model in which we assume that each 
component has a tendency to revert to the minimum of all components. 
Such a specification is useful to describe phenomena where 
each member in a population which is subjected to random noise mimicks the behaviour of the best performing member .

We show that the proposed dynamics generate co-integrated processes. 
 We characterize the model's asymptotic properties for the case of two populations and show a stabilizing effect on long term dynamics in simulation studies. An empirical study involving human survival data in different countries provides an example which confirms the occurrence of the phenomenon of 
reversion to the minimum in real data. 
\medskip

\end{abstract}

\section{Introduction}

  When  multivariate time series are used to describe the joint dynamics of stochastic processes, often an a priori assumption of a stabilizing mechanism is made. We expect, for example, that many economic variables will fluctuate randomly over time, but we do not find it plausible that they will diverge while doing so. This is because certain long-term equilibrium relationships are assumed to be present between such variables, despite the volatility we see in them over short time horizons. Many economic examples are now known (see for example \cite{EngleGranger1987}) and other fields in which such co-integration relationships are found include neuroscience (\cite{Ostergaard2017}) and gene differentiation in populations (\cite{Hoessjer2014}).

A possible model feature that implements such a stabilizing mechanism in multivariate time series is {\sl mean reversion}: the tendency of all components to drift towards a constant, which has the interpretation of the long term average of the series. Adding a mean reversion term to all increments of a multivariate discrete random walk makes the series second-order stationary. This also guarantees that the distance between any two components will be stationary.

 In this paper we propose a different mechanism: instead of making all components tend to 
a priori chosen constants, 
we impose that at every time step they move,
in expectation, 
towards the minimal value among all components.
When different components in the time series signify the value of a certain common variable among different groups, such an assumption can be interpreted as all groups mimicking the behaviour of the group which currently has the 'best' value\footnote{We consider reversion to the minimum in this paper and only consider biometric variables for which low values are deemed to be best. One could of course also apply our analysis in cases where there is a drift towards the maximum. Our choice here is motivated by the particular empirical example we give in the last section, where we consider multivariate time series for human mortality. }. 

 If knowledge about what is beneficial or detrimental for a group to minimize  a certain quantitative indicator is communicated  between different groups, 
and each group is capable of using this knowledge to their benefit, 
the best practices are implemented in all groups over time. The overall effect would be that the group which has achieved the best (i.e. the minimal) value would be 'followed' over time by those with worse achievements, as they 'learn' or 'mimic' what is beneficial. We will show that this effect will lead to a co-integration relationship between groups but also to an additional downward drift that no group would achieve on its own or if the minimum reversion effect would not be present. 
In that sense, learning from the best performer generates improvements for all that would not occur if individual groups were left to their own devices.

 This hypothesis of 'reversion to the minimum' in multivariate time series can be tested by specifying models with and without such an effect and 
applying the Bayesian Information Criterion 
to compare the goodness of fit. In this paper we provide such  specifications, analyse the properties of models where it is present, and give an empirical example of a specific time series where clear evidence for this effect is found. 

There are a number of 
models in the literature that are related to the one 
we propose in this paper. Systems of diffusion processes have been proposed in which the drift and diffusion coefficients of an individual process depend on the rank of that process within the system.
For example, \cite{Fernholz2002} and \cite{AtlasModel2005} introduce the Atlas model to describe the market capitalization of firms in equity markets, and \cite{Ichiba2013} apply their model to define optimal investment strategies. The Atlas model is applied by \cite{Sartoretti2013} to describe the dynamics of a swarm, and
\cite{Ruzmaikina2005} and \cite{Shkolnikov2009} consider the evolution of competing particle systems 
in discrete time and study the distribution of the gaps between the particles. 
\cite{Balazsa2014} introduced a continuous time Markov jump model for interacting particles with a jump rate which depends linearly on the distance of the particle from the center of mass of the whole group.

In contrast to the models proposed in the literature, we consider a system of discrete time processes under the assumption that 
the distribution of each component's stochastic increments 
does not 
depend on the rank of the component but on its distance to the minimum component.  This feature leads to multivariate time series with short term volatility and long term equilibrium relationships, while it gives a clear interpretation  of  the presence of the stabilizing mechanism in the long run.

The remainder of the paper is organised as follows. In section \ref{section:model} we introduce our model and analyse some of its theoretical properties.
Section \ref{section:empiricalstudy} provides a study in which we provide empirical evidence that a 'minimum reversion' effect can be found in human mortality data. Finally, we provide some conclusions and suggestions for further research in Section \ref{section:conclusion}.

\section{Minimum Reversion Model} \label{section:model}

\subsection{Specification } \label{subsection:minimalreversion}

Let  ${\cal T}=\{t_0,t_0+1,...,t_1\}$ for certain $t_0,t_1\in{\mathbb N}$ with $t_1>t_0$ and let $\{\kappa_t\}_{t\in{\cal T}}$ be a multivariate time series in ${\mathbb R}^C$ with components $\kappa_{t,c}$ with $c\in{\cal C}=\{1,2,...,C\}$ for a given \mbox{$C\in {\mathbb N}\setminus\{0,1\}$.}
We denote the minimum value among the different components at time $t \in {\cal T}$ by
\begin{equation} \label{eq_min}
m_t := \min_{c \in {\cal C}} \kappa_{t,c},
\end{equation}
and specify the following model for the dynamics of the components $\kappa_{t,c}$ with  $c \in {\cal C}$: 
\begin{eqnarray} \label{eq_seriesA}
 \kappa_{t+1,c}-\kappa_{t,c} &=& \mu_c + \zeta_c (\kappa_{t,c}-\kappa_{t-1,c}) 
+ \sigma_c Z_{t+1,c} 
+ \lambda_c( m_t-\kappa_{t,c}) .
\end{eqnarray}
 The $\mu_c\in{\mathbb R}$, $\zeta_c\in [0,1)$, 
$\lambda_c\in [0,1)$  
 and $\sigma_c\in{\mathbb R}^+$ are given constants. 
We define
\mbox{$Z_{t} = \left(Z_{t,1} , \ldots, Z_{t,C} \right)'$ }
and assume that the 
$\left\{Z_t\right\}_{t\in{\cal T}}$
are independent and identically distributed $C$-dimensional random variables with a multivariate Gaussian distribution.
More precisely, we assume that
\begin{eqnarray} 
 \label{eq_seriesB}
 Z_{t,c} &=& \rho_c W_{t,0} + \sqrt{1-\rho_c^2} W_{t,c},
\end{eqnarray}
for independent standard Gaussian variables $\{W_{t,c}\}_{c \in {\cal C}\cup\{0\}, t \in {\cal T}}$,
which creates a correlation structure parameterized by the $C$ constants 
$\rho_c\in\, (-1,1)$.  

 The constant drift $\mu_c$, the first order autoregression coefficient for the differenced series $\zeta_c$ and the Gaussian increment $\sigma_c Z_{t+1,c}$ are standard features in time series modelling. Our innovation concerns the $\lambda_c$ parameters which quantify the 'learning effect'. 
To discuss the properties of the newly proposed term in (\ref{eq_seriesA}) we first consider the model with $\mu_c=\zeta_c = 0$ and $\lambda_c>0$. 
This introduces a downward drift $\lambda_c( m_t-\kappa_{t,c})$ for $\kappa_c$ when $c$ is not the component with the lowest current value. 

The specification in (\ref{eq_seriesA}) with $\mu_c=\zeta_c = 0$ and $\lambda_c>0$ also creates a downward trend in the minimum process $m_t$ defined in (\ref{eq_min}). 
To form some intuition for this result we consider the case where $\rho_c=0$ for all $c$, that is, the $Z_{t,c}$ are independent for any fixed $t$. In that case we have
$${\mathbb P}(m_{t+1} \leq a \mid \kappa_t)=1- \prod_{c \in {\cal C}} 
\Phi \left(\frac{(1-\lambda_c) \kappa_{t,c} -a+\lambda_c m_t }{\sigma_c} \right).$$ Substituting $a=m_t$ gives an expression for ${\mathbb P}(m_{t+1} \leq m_t \mid \kappa_{t})$
and since $\kappa_{t,c}\geq m_t$ for all $c$ and $\kappa_{t,c} = m_t$ for at least one $c$, this probability must be strictly greater 
 than $\sfrac12$ and smaller than or equal to $1-2^{-C}$. This establishes that the probability of a downward movement is always 
greater 
than $\frac12$ and that it is 
largest 
 when all components attain the same value, that is, $\kappa_{t,c}=m_t$ for all $c\in\cal C$. The probability is increasing with the number of components $C$.
 
Individual components are thus not stationary but they turn out to be co-integrated, as the following result shows.

\begin{proposition} 
If all processes  in (\ref{eq_seriesA}) have a common minimum reversion parameter and a common drift and there is no autoregressive term,  so $\mu_c=\mu$, \mbox{$\lambda_c=\lambda>0$} and $\zeta_c = 0$ for all $c\in {\cal C}$, then
the processes $\{\kappa_{\cdot,c}\}_{c\in {\cal C}}$ are co-integrated.
\end{proposition}

\begin{proof} 
 Fix a $c^* \in {\cal C}$ and define $\tilde{\kappa}_{t,c} := \kappa_{t,c}-\kappa_{t,c^*}$ for any $c \in {\cal C}$. We then find for any $c \neq c^*$ 
\begin{eqnarray*}
\tilde{\kappa}_{t,c} 
&=& (1-\lambda)(\kappa_{t-1,c}-\kappa_{t-1,c^*})
+ \tilde{Z}_{t} \\
&=& (1-\lambda)\tilde{\kappa}_{t-1,c}
+ \tilde{Z}_{t},\hspace*{3cm}
\tilde{Z}_{t} \ = \  \sigma_{c}Z_{t,c}-\sigma_{c^*}Z_{t,c^*}.
\end{eqnarray*}
Since $0 < \lambda \leq 1$ we obtain that $\tilde{\kappa}_{t,c}$ is a stationary AR(1) process for all $c \neq c^*$. 

Furthermore, we find that
$
m_t =  \min_{c\in {\cal C}} (\kappa_{t,c^*}+\tilde{\kappa}_{t,c})\  \ = \ \kappa_{t,c^*} +  \min_{c\in {\cal C}}\tilde{\kappa}_{t,c}
$
and therefore 
\begin{eqnarray*}
\Delta \kappa_{t+1,c^*} \ := \  \kappa_{t+1,c^*} - \kappa_{t,c^*} 
&=& \mu_{c^*} +  \lambda( m_t-\kappa_{t,c^*}) 
+ \sigma_{c^*} Z_{t+1,c^*}\\
 &=& \lambda \min_{c\in {\cal C}}\tilde{\kappa}_{t,c}
+ \mu_{c^*} +  \sigma_{c^*} Z_{t+1,c^*}.
\end{eqnarray*}
The first term in the last expression is a minimum over stationary processes and the other terms are stationary too, hence $\Delta \kappa_{t,c^*}$ is stationary.
\end{proof}

\subsection{A analysis of the two dimensional case}

To characterize the dynamics of minimum reversion, and to show that it induces a downward drift even for the case where the drift parameters $\mu_c$ are taken to be zero, we look at the simplest non-trivial case in two dimensions ($C=2)$. This means the dynamics of the time series become:
\begin{eqnarray} \label{eq_series2}
 \kappa_{t+1,1}-\kappa_{t,1} &=&
 \lambda( \min_{c\in\{1,2\}}\kappa_{t,c}-\kappa_{t,1}) 
+ \sigma_{1}Z_{t+1,1} +\mu,\\
 \kappa_{t+1,2}-\kappa_{t,2} &=&
 \lambda( \min_{c\in\{1,2\}}\kappa_{t,c}-\kappa_{t,2}) 
+ \sigma_{2}Z_{t+1,2} +\mu. \label{eq_series2a}
\end{eqnarray}
Note that we take  $\lambda_1=\lambda_2:=\lambda$
and $\mu_1=\mu_2:=\mu$
 for the analysis in this section and we assume that the correlation between the i.i.d. Gaussian variables $Z_{t+1,1}$ and $Z_{t+1,2}$  equals $\rho$ for all $t$.

We define\footnote{If $\kappa_{t,1}=\kappa_{t,2}$ we let $m_t=\kappa_{t,1}$ and $M_t=\kappa_{t,2}$, but this event has zero probability for all $t\in {\cal T}$.} 
\[ m_t=\min_{c\in\{1,2\}}\kappa_{t,c} \  \mbox{ and }   \ M_t=\max_{c\in\{1,2\}}\kappa_{t,c}.\]

\begin{proposition} \label{prop_TwoPop}
For the model \eqref{eq_series2}-\eqref{eq_series2a} with $\lambda_1=\lambda_2:=\lambda$ we have that 
\begin{equation}\label{eq:mdrift}
\lim_{t\to\infty}\ {\mathbb E}[M_{t+1}-M_t] \ = \ 
\lim_{t\to\infty}\ {\mathbb E}[m_{t+1}-m_t] \ = \ \mu\ -\ s\sqrt{\sfrac{\lambda}{2\pi(2-\lambda)}}
\end{equation}
and
\begin{equation}\label{eq:mMspread}
\lim_{t\to\infty}\ {\mathbb E}[M_t-m_t] \ = \  s\sqrt{\sfrac{2}{\lambda\pi(2-\lambda)}}
\end{equation}
with  $s=\sqrt{\sigma_1^2+\sigma_2^2-2\rho\sigma_1\sigma_2}$.
\end{proposition}

\begin{proof}
Let the indices of the minimum and maximum be denoted by $i_t=\arg\min_{c\in\{1,2\}}\kappa_{t,c}$ and
$I_t=\arg\max_{c\in\{1,2\}}\kappa_{t,c}$. We first determine the distribution of the minimum and maximum\footnote{Since the variables are Gaussian this distribution can be obtained directly as well, see \cite{Roberts1966}.}.
Substituting in \eqref{eq_series2}-\eqref{eq_series2a} gives, for
$$E_{t+1} \ = \ M_t + \lambda(m_t-M_t) +  \sigma_{I_t}Z_{t+1,I_t} - (m_t + \sigma_{i_t}Z_{t+1,{i_t}}),$$ that
\begin{eqnarray*}
m_{t+1} &=& m_t + \sigma_{i_t}Z_{t+1,i_t} + \mu \ + \qquad\qquad\qquad\ \ \ E_{t+1}{\bf 1}_{E_{t+1}<0 },\\
M_{t+1} &=& M_t + \sigma_{I_t}Z_{t+1,I_t} + \mu + \lambda(m_t-M_t)  \ - \  E_{t+1}{\bf 1}_{E_{t+1}\geq 0 }.
\end{eqnarray*}
We find
\begin{eqnarray*}
{\mathbb E}[m_{t+1}|(m_t,M_t)] 
&=& m_t + \mu +  ( 1-\lambda)(M_t-m_t) {\mathbb P}(E_{t+1}<0) \\
&& + {\mathbb E}[ (\sigma_{I_t}Z_{t+1,I_t} -\sigma_{i_t}Z_{t+1,i_t} ){\bf 1}_{E_{t+1}<0}\ | \ (m_t,M_t)]\\
{\mathbb E}[M_{t+1}|(m_t,M_t)] 
&=& M_t + \mu  + \lambda(m_t-M_t)  - ( 1-\lambda)(M_t-m_t){\mathbb P}(E_{t+1}>0)\\&& - {\mathbb E}[ (\sigma_{I_t}Z_{t+1,I_t} -\sigma_{i_t}Z_{t+1,i_t} ){\bf 1}_{E_{t+1}>0}\ | \ (m_t,M_t)].
\end{eqnarray*}
The term $H:=\sigma_{I_t}Z_{t+1,I_t} -\sigma_{i_t}Z_{t+1,i_t}$ is independent of $(m_t,M_t)$ and has a Gaussian distribution with mean zero and variance $s^2\ :=\ {\sigma_{i_t}^2+\sigma_{I_t}^2}-2\rho\sigma_{i_t}\sigma_{I_t}\ = \ {\sigma_1^2+\sigma_2^2}-2\rho\sigma_1\sigma_2$ so ${\mathbb E}[H{\bf 1}_{H<a }]=-se^{-(a/s)^2/2}/\sqrt{2\pi}$ for all $a\in\mathbb R$. We thus find that ${\mathbb E}[m_{t+1}|(m_t,M_t)] $ equals\footnote{Here and in the sequel we use $\Phi$ to denote the cumulative probability function for a standard Gaussian random variable.}
\begin{eqnarray}\nonumber
\ &\ & m_t + \mu +( 1-\lambda)(M_t-m_t)\Phi(-\textstyle\frac{ (1-\lambda)(M_t-m_t)}{s})-s e^{-\textstyle\frac12 ( (1-\lambda)(M_t-m_t)/s )^2}/ \sqrt{2\pi}\\ 
&=& m_t +  \mu +  s (\ \tilde{D}_t \Phi(-\tilde{D}_t) - e^{-\frac12 \tilde{D}_t^2}/ \sqrt{2\pi} \ )
\ = \  m_t +   \mu + s\, f(\tilde{D}_t),
\label{eq:incrementminimum}
\end{eqnarray}
if we define the shorthand notation $D_t=M_t-m_t$ and $\tilde{D}_t=(1-\lambda)D_t/s$. The function $x\to f(x)=x\Phi(-x)-\exp(-x^2/2)/\sqrt{2\pi}$ increases monotonically from its minimum \mbox{$f(0)=-1/\sqrt{2\pi}=-0.3989$} for $x=0$ towards zero for $x=\infty$. Likewise,
\begin{eqnarray*}
{\mathbb E}[M_{t+1}|(m_t,M_t)] &=& M_t +   \mu  + \lambda(m_t-M_t) -   s\, f(\tilde{D}_t).
\end{eqnarray*}
Since 
$
D_{t+1}=M_{t+1}-m_{t+1}=|(1-\lambda)D_t+s\tilde{H}|$  for a  $\tilde H$ with the standard normal distribution,
we must have that  $\tilde{D}_{t+1}=|(1-\lambda)^2D_t/s+(1-\lambda)s{\tilde H}/s|=
(1-\lambda)\ |\tilde{D}_t+{\tilde H}|$.

Let $V(x)=|x|+1$, and define $\Delta V(x)={\mathbb E}[V(\tilde{D}_{t+1})-V(\tilde{D}_t)\mid \tilde{D}_t=x]$, then we have for  $\beta=\lambda$ and $b=(1-\lambda){\mathbb E}\, |H|$ that $ \Delta V(x)\leq -\beta V(x)+b$ which shows that the process $\tilde{D}$ satisfies the geometric ergodicity conditions in Theorem 15.0.1 of \cite{MeynTweedieBook} on the domain ${\mathbb R}^+$. This proves that for any choice of $\tilde{D}_0\geq 0$ the distribution function $G_t$ of $\tilde{D}_t$ converges to a stationary distribution $G$. This  implies  that for $x\geq 0$
the function $G$ should satisfy
$$G(x)={\mathbb P}(\tilde{D}_t\leq x)={\mathbb P}((1-\lambda)\ |\tilde{D}_t+\tilde{H}|\leq x) \ = \int_0^\infty [N(\sfrac{x}{1-\lambda}-y)-N(-\sfrac{x}{1-\lambda}-y)]dG(y).$$
 The solution to this equation is $G(x)={\bf 1}_{ x\geq 0}\int_0^x 2\exp({-\sfrac12(u/\tilde{\sigma})^2})/( \tilde{\sigma}\sqrt{2\pi})du$ with the parameter {$\tilde{\sigma}=(1-\lambda)/\sqrt{\lambda(2-\lambda)}$},  i.e. the stationary distribution of $\tilde D$ is the distribution of $|X|$ if $X\sim N(0,\tilde{\sigma}^2)$. The  stationary distribution of $M_t-m_t=D_t=s\tilde{D}_t/(1-\lambda)$ is therefore the  distribution of $s|X|/(1-\lambda)$, which gives \eqref{eq:mMspread}.

 Due to \eqref{eq:incrementminimum}, the expected change in $m_t$ under the stationary distribution for $\tilde{D}_t$  equals $\mu$ plus the effect of the reversion to the minimum. This effect was shown above to be
$ s{\mathbb E}[f(\tilde{D}_t)]=s\int_0^\infty f(x)dG(x)=s(\tilde{\sigma} - \sqrt{1+\tilde{\sigma}^2})/\sqrt{2\pi}$
which gives the second equality in \eqref{eq:mdrift} after some rewriting. The first one follows from the fact that $\tilde{D}$ converges to a stationary distribution. 
\end{proof}

\medskip

If we choose certain parameter values, the asymptotic drifts in proposition \ref{prop_TwoPop} can 
also be 
estimated 
by Monte Carlo simulations. We generated $10^5$ paths to approximate the stationary distributions and used these to estimate 
the limit of  the expectation of the extra drift  (Table \ref{table1}) and the limit of the expectation of the difference
between the minimum and maximum (Table \ref{table2}) for $s=1$, $\mu=0$ and in the absence of correlation between the increments for the two time series. 
Simulations also allow us to estimate the extra drift generated by the minimum reversion term in more than two dimensions, i.e. for $|{\cal C}|>2$ and we show some examples Tables \ref{table1} and \ref{table2}. 
As expected, we find that both the drift generated by the minimum reversion and the difference between the minimum and maximum increases with the population size $C$. The former increases and the latter decreases when the strength of the reversion to the minimum, which is determined by the parameter $\lambda$, increases.

\begin{table}[t]
\tbl{Generated drift  for different minimum reversion parameters $\lambda$ according to \eqref{eq:mdrift}  (second column) and using  $10^5$ simulations (third to last column).}
{
\begin{tabular}{crrrrrrr}
\ \\
                               & Exact & Simulation &&&&& \\ 
    $\lambda$            & $|{\cal C}|=2$ & $|{\cal C}|=2$ &  $|{\cal C}|=3$
    &  $|{\cal C}|=4$  &  $|{\cal C}|=8$  &  $|{\cal C}|=16$  &  $|{\cal C}|=32$ \\ 
   0.0125 & -0.0447 & -0.0448 & -0.0671 & -0.0817 & -0.1129 & -0.1401 & -0.1641 \\
    0.025 & -0.0635 & -0.0635 & -0.0952 & -0.1158 & -0.1602 & -0.1987 & -0.2329 \\
    0.05  & -0.0903 & -0.0903 & -0.1355 & -0.1649 & -0.2280 & -0.2828 & -0.3314 \\
    0.1   & -0.1294 & -0.1294 & -0.1942 & -0.2362 & -0.3266 & -0.4052 & -0.4748 \\
    0.2   & -0.1881 & -0.1881 & -0.2821 & -0.3432 & -0.4745 & -0.5887 & -0.6899 \\
    0.4   & -0.2821 & -0.2821 & -0.4231 & -0.5147 & -0.7118 & -0.8830 & -1.0348 \\
\end{tabular}
}
\label{table1}
\end{table}

\begin{table}[t]
\tbl{Expectation of the stationary distribution for $M_t-m_t$  for different minimum reversion parameters $\lambda$ according to \eqref{eq:mMspread} (second column) and using  $10^5$ simulations (third to last column).}
{
\begin{tabular}{crrrrrrr} \ \\
                               & Exact & Simulation &&&&& \\
    $\lambda$            & $|{\cal C}|=2$ & $|{\cal C}|=2$ &  $|{\cal C}|=3$
    &  $|{\cal C}|=4$  &  $|{\cal C}|=8$  &  $|{\cal C}|=16$  &  $|{\cal C}|=32$ \\ 
   0.0125 & 7.1589 & 7.1593 & 10.7386 & 13.0627 & 18.0644 & 22.4106 & 26.2618 \\
    0.025 & 5.0781 & 5.0774 & 7.6185 & 9.2656 & 12.8138 & 15.8962 & 18.6282 \\
    0.05  & 3.6137 & 3.6132 & 5.4202 & 6.5931 & 9.1178 & 11.3112 & 13.2565 \\
    0.1   & 2.5887 & 2.5886 & 3.8831 & 4.7229 & 6.5316 & 8.1031 & 9.4965 \\
    0.2   & 1.8806 & 1.8806 & 2.8209 & 3.4313 & 4.7454 & 5.8866 & 6.8988 \\
    0.4   & 1.4105 & 1.4104 & 2.1156 & 2.5735 & 3.5591 & 4.415 & 5.1742 \\
\end{tabular}
}
\label{table2}
\end{table}

\section{Evidence for a Learning Effect in Mortality Rates}\label{section:empiricalstudy}

In this section we use the time series model in (\ref{eq_seriesA}) to model
mortality rates  
 in several countries.
 The implicit assumption is that a population with the high mortality rate "copies" the behaviour of individuals in the population with low mortality. Such a model is consistent with the spread of medical advances and changes in behaviour like a reduction in smoking prevalence.

\subsection{Modelling Mortality - The Common Age Effect Model}

 To obtain the time series, we first fit a stochastic mortality model to observed death counts.
We assume that the number of deaths, $D_{xtc}$, in population $c\in{\cal C}$ at age $x\in{\cal X}$ in calendar year $t\in{\cal T}$ is a random variable with a Poisson distribution, that is, 
\begin{equation} \label{eq_Poisson}
 D_{xtc} \sim \mbox{Pois}\left(\mu_{xtc} E_{xtc} \right) 
\end{equation}
where $\mu_{xtc}$ is the hazard rate (also known as the "force of mortality" in the literature) and $E_{xtc}$ refers to the central 
exposure 
 to risk.

 We then use a stochastic model for the force of mortality $\mu_{xtc}$ that incorporates the population-specific time series $\{\kappa_t\}_{t\in{\cal T}}$ defined in
\eqref{eq_seriesA}. 
As we are modelling the mortality in multiple populations simultaneously and wish to make our model suitable for a wide age range, we use a modification of the Lee-Carter model (\cite{LeeCarter:92}) with common age effects, 
as suggested by \citet{Kleinow:2015}:
\begin{equation} \label{eq_CAEmodel}
\log \mu_{xtc} = \alpha_{x}  + \beta_{x}  \kappa_{t,c} 
\end{equation} 
where the age effects $\alpha_{x}$ and $\beta_{x}$ do not depend on the population $c$. 
Having age effects that are common to all populations ensures that the individual components $\kappa_{t,c}$ (known as "period effects") are comparable across populations as they are rescaled by a common vector $\beta$ and shifted by a common vector $\alpha$. 
The parameters in (\ref{eq_CAEmodel}) are not identifiable since 
\begin{equation*}
\alpha_{x}  + \beta_{x}  \kappa_{t,c} =
\tilde\alpha_{x}  + \tilde\beta_{x}  \tilde\kappa_{t,c}\quad\mbox{ when }\quad
\tilde\alpha_{x} = \alpha_{x} - \sfrac{K_1}{K_2}\beta_x,\quad
\tilde\beta_{x} = \sfrac{\beta_{x}}{K_2},\quad
\tilde\kappa_{t,c} = K_1+K_2\, \kappa_{t,c} .
\end{equation*}
for any real numbers $K_1$ and $K_2 \neq 0$. 
To identify a unique set of parameters, 
 we impose the following constraints on the parameter vectors $\alpha$ and $\beta$:
\begin{equation*} 
\alpha_{x_r} = 0 \mbox{ and } \beta_{x_r} = 1
\end{equation*}
for a fixed reference age $x_r \in {\cal X}$. 
Applying those two constraints means that $\log \mu_{xtc} = \kappa_{t,c}$ for $x=x_r$ in every population $c \in {\cal C}$. In other words, we can interpret the period effect $\kappa_{t,c}$ as the fitted log mortality rate at the reference age $x_r$ in population $c$. 
For our empirical analysis we choose the reference age $x_r = 70$.

We estimate the age effects $\alpha$ and $\beta$ and the period effects $\kappa$ using the maximum likelihood method based on the Poisson model in (\ref{eq_Poisson}) for observed deaths counts $D_{xtc}$ and exposure values $E_{xtc}$.

\begin{figure}[t]
\begin{center}
\infig{0.45}{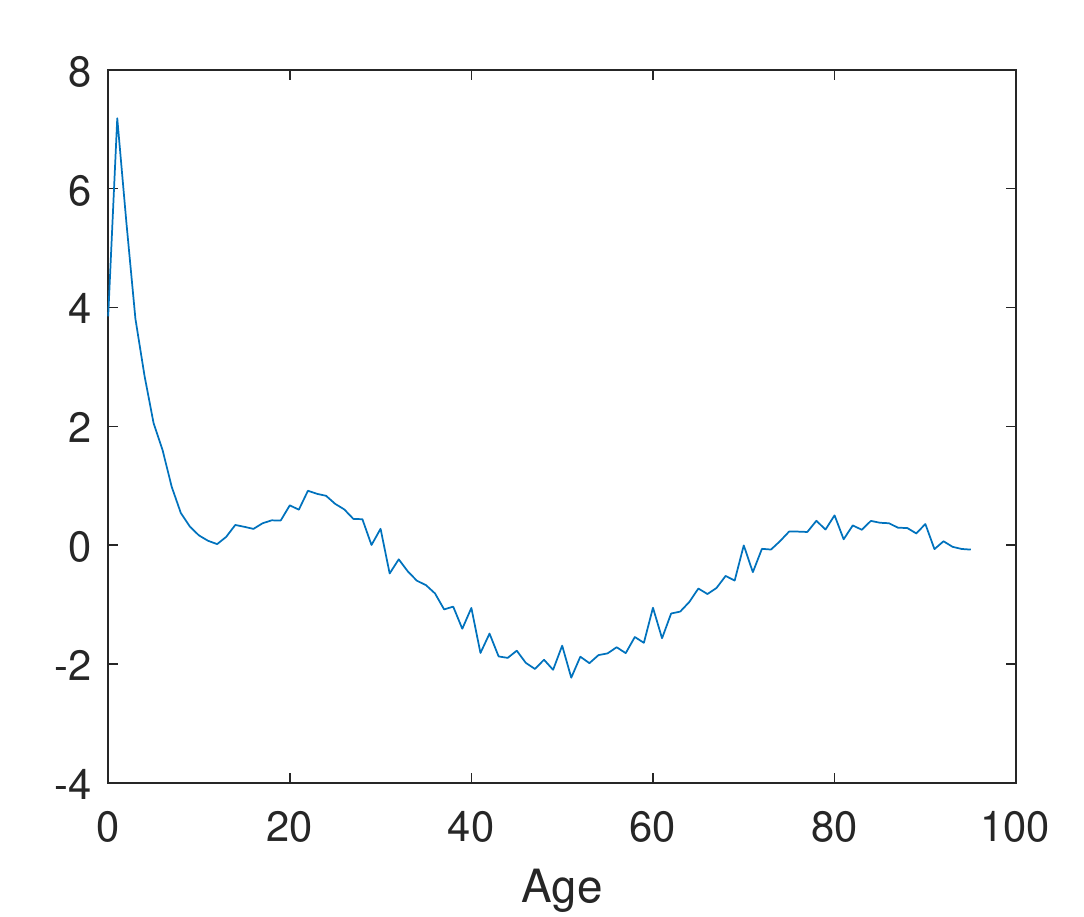} 
\infig{0.45}{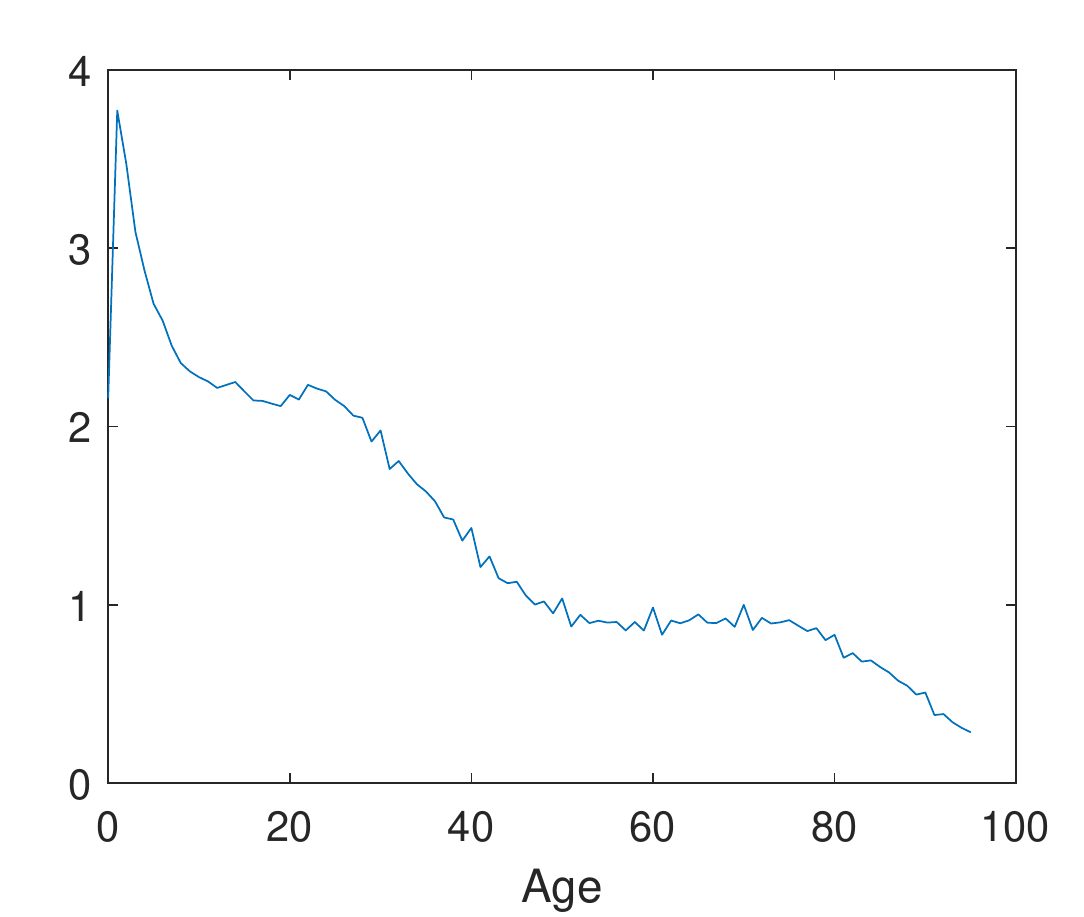} 
\end{center}
\caption{\label{fig_ageEffects} Estimated values of $\alpha_x$ (left) and $\beta_x$ (right) for the mortality rates of females based on data for the period 1921--2011. We can clearly see the phenomenon known as "age heaping": when people report the death of a person without knowing their exact age, they often report the age to be a multiple of $10$.}
\end{figure}

\subsection{Data} \label{Sec_Data}

 The empirical death counts and exposure data have been obtained from the Human Mortality Database (HMD), see \cite{HMD}. We consider data for ages ${\cal X}=\{0,1,\ldots,95\}$ and for two ranges of calendar years: ${\cal T}=\{1951,\ldots,2011\}$ 
or 
${\cal T}=\{1921,\ldots,2011\}$. Table \ref{table_Countries} shows a list of countries included in our study together with the HMD country codes. As indicated in the table, there are some countries for which data from 1921 are not available and those countries have been excluded from our empirical study for that range or years. 
\begin{table}[htb]
\small
\begin{center}
\caption{\label{table_Countries} List of countries included in our study. For all countries mortality data for males and females are considered for calendar years 1951--2011. The table also indicates which countries are included in our study for the longer observation period 1921--2011. Note that for France and England \& Wales the total populations are considered rather than the civilian population.}
\begin{tabular}{lcclcc}\ \\
Country & HMD Code & 1921--2011 & Country & HMD Code & 1921--2011  \\ \ \\
The Netherlands & NLD & \checkmark & 
Sweden & SWE & \checkmark \\
Denmark & DNK & \checkmark &
Belgium & BEL & \checkmark \\
Finland & FIN & \checkmark &
England \& Wales & GBRTENW& \checkmark \\
France & FRATNP & \checkmark  &
Switzerland & CHE & \checkmark  \\
Australia & AUS & \checkmark  &
Italy & ITA & \checkmark  \\
Austria & AUT & &
Ireland & IRL & \\
Norway & NOR & \checkmark  &
Japan & JAP & \\
Canada &CAN & \checkmark  &
New Zealand & NZL\_NP & \\
Portugal & PRT & &
Spain &ESP & \checkmark  \\
USA & USA & &
Iceland & ISL & \checkmark 
\end{tabular}
\end{center}
\end{table}

\subsection{Empirical Results}

Figure \ref{fig_ageEffects} shows the estimated values of the age effects $\alpha_x$ and $\beta_x$ obtained from data for females based on the observation period 1921--2011. 
The age effects for the other data sets are not shown since they all have a very similar shape.
To illustrate our data 
and the effect of the minimum reversion effect,
Figure \ref{fig_periodEffects} shows the estimated values of $\kappa_{t,c}$ for females for the years 1921--2011 in the dataset, together with two projected scenarios based on the model in (\ref{eq_seriesA}). On the lefthand side the projection includes the minimum reversion effect based on the estimated parameter $\lambda$ and on the righthand side we set $\lambda = 0$. 

To investigate the significance of the learning effect, we compare the Bayesian Information Criterion (BIC) for models without minimum reversion ($\lambda = 0$) with the more general model in (\ref{eq_seriesA}). If $K$ denotes the number of parameters and $n = |{\cal C}||{\cal T}|$ is the total number of observations across all populations and all years, the BIC value equals  $K\log n - 2 \log L$, where $L$ is the maximum value of the likelihood function $L$ for the model in (\ref{eq_seriesA}). This means 
\[ 2 \log L = - n \log| \Sigma | - \sum_{t=t_0+2}^{t_1} \left(\Delta \kappa_t - \eta_{t-1} \right)' \Sigma^{-1} \left(\Delta \kappa_t - \eta_{t-1} \right) + \Gamma, \] 
where $\Gamma$ is a constant which does not depend on the parameters that must be estimated, $\Delta \kappa_{t}$ and $\eta_{t}$ are vectors with components
\[
\Delta \kappa_{t,c} = \kappa_{t,c}-\kappa_{t-1,c} \mbox{ and }
 \eta_{t,c} = \mu_c + \zeta_c \Delta \kappa_{t,c} + \lambda_c( m_t-\kappa_{t,c}) ,
\]
and $\Sigma$ is the covariance matrix of $\Delta \kappa_t$ so $\Sigma_{ij} = \sigma_i \sigma_j \rho_i \rho_j$ for $i\neq j$ and $\Sigma_{ii} = \sigma_i^2$.

We determined the BIC values for different specifications in which parameter values may or may not be constrained to be the same for all different groups in ${\cal C}$.  
Based on this analysis, the parameters $\sigma_c$ and $\rho_c$ are taken population-specific while the parameters $\mu_c=\mu$, $\lambda_c=\lambda$ and $\zeta_c=\zeta$ are common to all populations. The obtained BIC values and parameter estimates\footnote{We do not show the values $\rho_c$
 and $\sigma_c$ for every group $c$ in $\cal C$ but only report $\bar\rho$, the average of $\rho_c$ over all groups $c$, and ${\bar \sigma}$, which is defined by the requirement that ${\bar\sigma}^2$ is the average of $\sigma^2_c$ over all groups $c$.   
Estimated parameter values for all countries, genders and data periods are available upon request.} are shown in Table \ref{tableBIC} for four data sets.

\begin{table}\begin{center}
\small
\caption{Goodness of fit and parameter values of the model in (\ref{eq_seriesA}) without and with minimum reversion, estimated for different data sets.}
\label{tableBIC}
\begin{tabular}{lcccccccc} \ \\
Model & $\log L$ & $K$ & BIC & ${\hat \mu}$   & ${\hat \lambda}$ & ${\hat \zeta}$ & ${\bar \sigma}$ &  ${\bar \rho}$ \\ 
\multicolumn{8}{l}{Females in calendar years 1921--2011}\\
$\lambda = 0$    & 3886.16 & 30 & -7558.49 & -0.0224 &        & -0.2998 & 0.0338 & 0.4996\\ 
$\lambda \neq 0$ & 3890.61 & 31 & -7560.27 & -0.0194 & 0.0191 & -0.2895 & 0.0336 & 0.4977\\ 
\multicolumn{8}{l}{Males in calendar years 1921--2011}\\
$\lambda = 0$    & 4151.49 & 30 & -8089.15 & -0.0150 &        & -0.1890 & 0.0297 & 0.5207\\ 
$\lambda \neq 0$ & 4155.65 & 31 & -8090.35 & -0.0131 & 0.0185 & -0.1804 & 0.0295 & 0.5224\\ 
\multicolumn{8}{l}{Females in calendar years 1951--2011}\\
$\lambda = 0$    & 3854.23 & 42 & -7411.38 & -0.0210 &        & -0.3414 & 0.0288 & 0.4717\\ 
$\lambda \neq 0$ & 3861.62 & 43 & -7419.08 & -0.0161 & 0.0177 & -0.3394 & 0.0285 & 0.4641\\ 
\multicolumn{8}{l}{Males in calendar years 1951--2011}\\
$\lambda = 0$    & 4100.51 & 42 & -7903.94 & -0.0246 &        & -0.3320 & 0.0252 & 0.5845\\ 
$\lambda \neq 0$ & 4105.98 & 43 & -7907.81 & -0.0192 & 0.0167 & -0.3219 & 0.0250 & 0.5832\\ 
\end{tabular}%
\end{center}\end{table}

We notice that the BIC values always improve when $\lambda$ is not restricted to be zero which shows that the learning effect adds to the goodness of fit of the model. We also observe that the estimated drift parameter $\hat \mu$ is reduced when minimum reversion is included in the model, from which we conclude that some of the mortality improvements in the populations are driven by learning effects from others who have lower mortality rates.

\section{Conclusions and Further Research} \label{section:conclusion}

\begin{figure}[b]
\begin{center}
\infig{0.45}{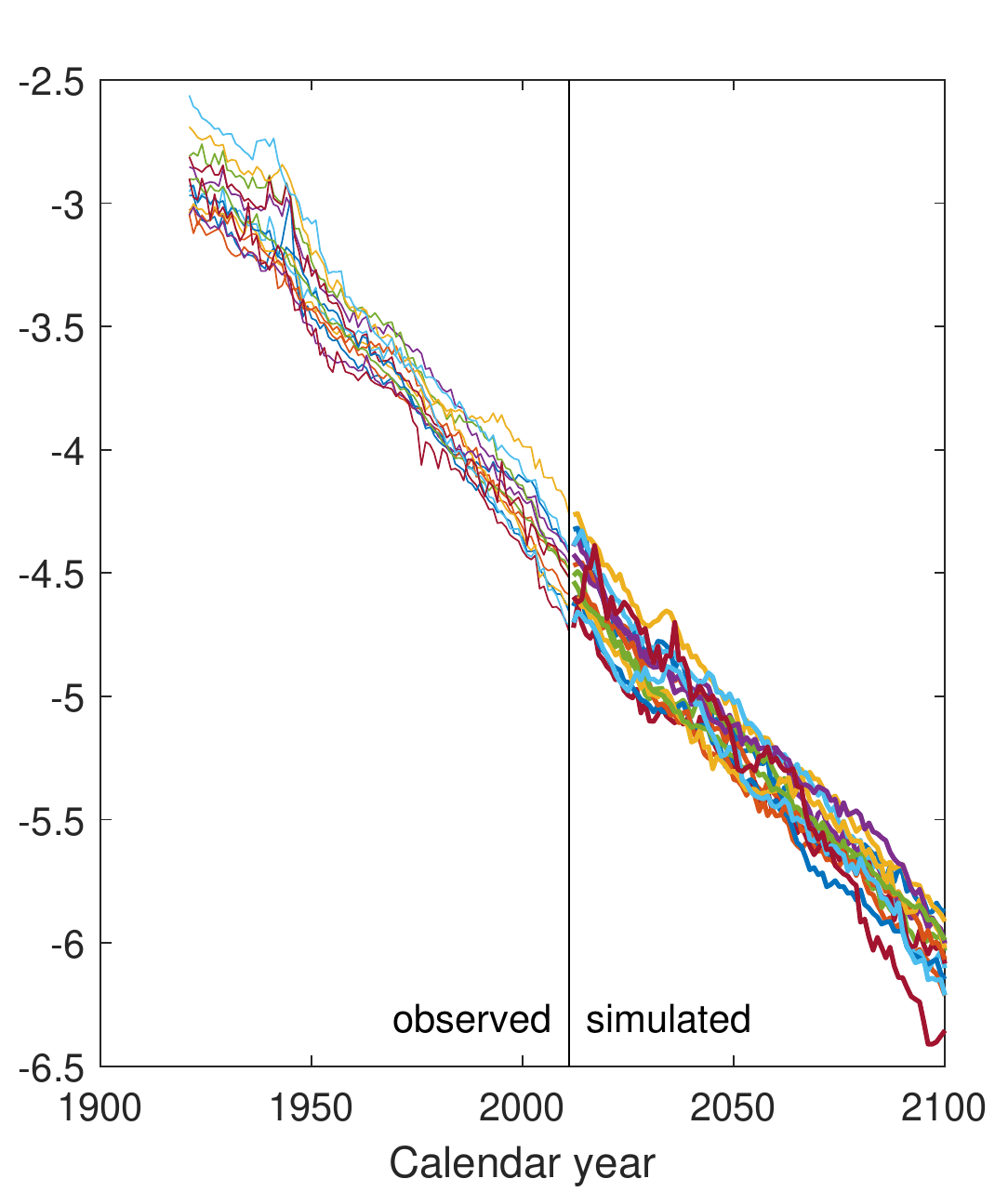} 
\infig{0.45}{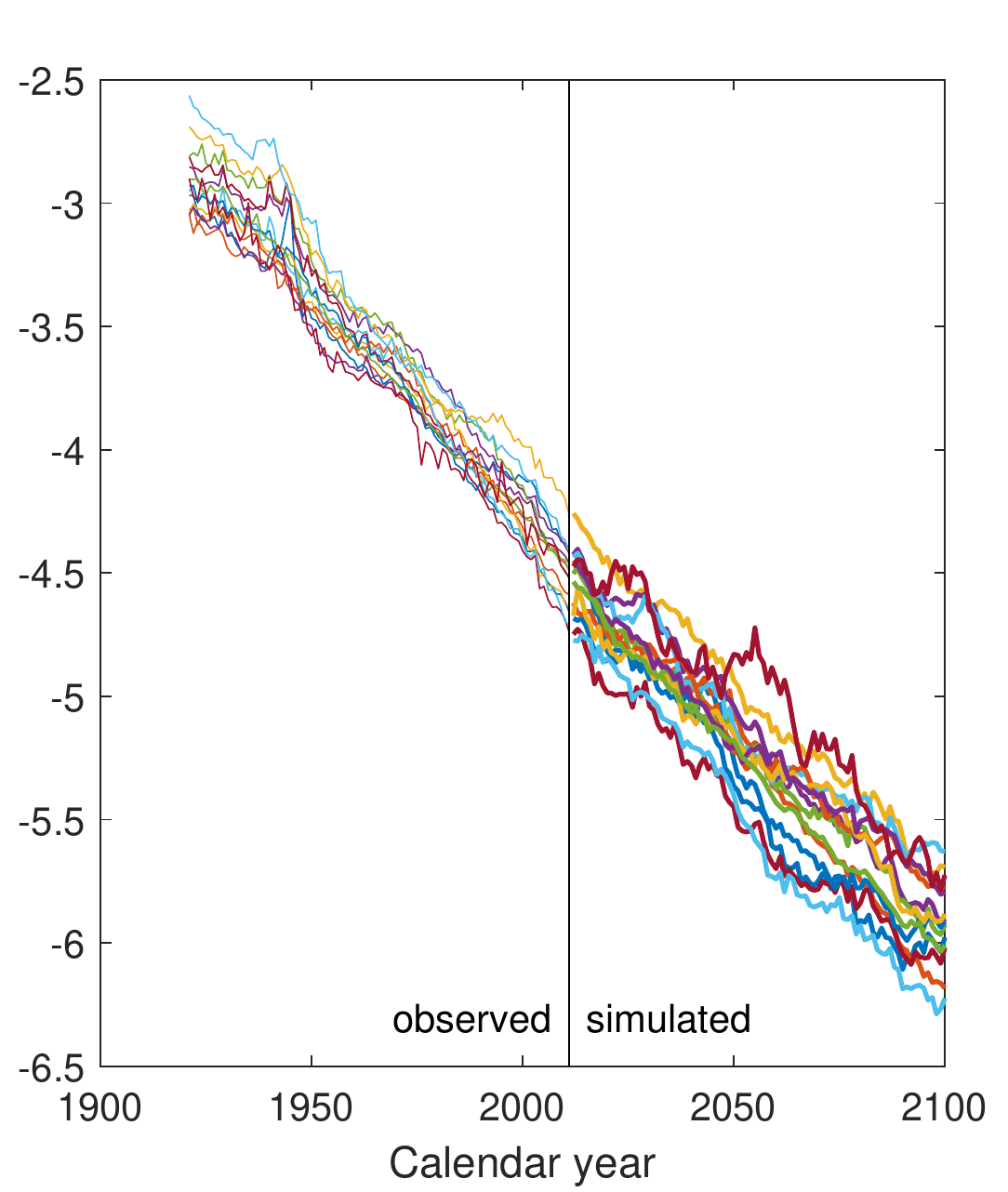} 
\end{center}
\caption{\label{fig_periodEffects} Estimated values of $\kappa_{t,c}$ for the mortality rates of females based on data for the period 1921--2011. The simulated scenarios are based on the full model in (\ref{eq_seriesA}) (left) and the nested model with $\lambda = 0$ (right).}
\end{figure}

We have shown that model specifications that include a minimum reversion term resulted in better BIC values than specifications without such a term when they were fitted to several datasets from the human mortality database. Visual inspection of the projections generated by the time series with minimum reversion in \ref{fig_periodEffects} also shows a clear improvement. This testifies to the usefulness of the incorporation of a "learning effect" in the time series.

 Several extensions of the proposed model are possible. In particular, the minimum in (\ref{eq_seriesA}) could be replaced by other rank statistics. Another possible direction of research  
is the study of 
a continuous time version of minimum reversion, where the drift term of a diffusion process in multiple dimensions is a function of the distance of the process to the minimum of its components. 
The properties of these processes 
could then be compared 
to the rank-based diffusion processes for particle systems
mentioned in the Introduction of this paper.

\bibliographystyle{biometrika}
\bibliography{MultiPopModels}

\end{document}